\documentclass[a4paper,cleveref,autoref]{lipics-v2021}

\usepackage{prettyref}

\newcommand{\rref}[2][]{\prettyref{#2}}
\newrefformat{sec}{Section\,\ref{#1}}
\newrefformat{app}{Appendix\,\ref{#1}}
\newrefformat{def}{Def.\,\ref{#1}}
\newrefformat{thm}{Theorem\,\ref{#1}}
\newrefformat{prop}{Proposition\,\ref{#1}}
\newrefformat{lem}{Lemma\,\ref{#1}}
\newrefformat{cor}{Corollary\,\ref{#1}}
\newrefformat{rem}{Remark\,\ref{#1}}
\newrefformat{ex}{Example\,\ref{#1}}
\newrefformat{tab}{Table\,\ref{#1}}
\newrefformat{fig}{Fig.\,\ref{#1}}
\newrefformat{case}{case\,\ref{#1}}
\newrefformat{foot}{Footnote\,\ref{#1}}

\usepackage{amsthm}

\usepackage{amsmath}
\usepackage{amssymb}
\usepackage{mathtools}
\usepackage{tensor}
\usepackage{bm}
\usepackage{tikz}
\usepackage{tikz-cd}
\usepackage{graphicx}
\usepackage{stmaryrd}
\usepackage{svg}

\usepackage{xcolor}
\usepackage{float}
\usepackage[prefixflatinterpret,prefixflatbracketmodalinterpret,prefixflatsetinterpret,sidenotecalculus,longseqcontext]{logic}
\usepackage[prefixflatinterpret,prefixflatbracketmodalinterpret,differentialdL,simplenames]{dL}

\graphicspath{ {./Images/} }

\usepackage{xspace}

\definecolor{semblue}{rgb}{0,0,0.7}
\definecolor{vgreen}{rgb}{.1,.5,0}%
\definecolor{vdarkgreen}{rgb}{.06,.3,0}%
\definecolor{vred}{rgb}{.7,0,0}%
\definecolor{vblue}{rgb}{.1,.15,.62}%
\definecolor{vgray}{rgb}{.35,.35,.35}
\definecolor{darkishgray}{rgb}{.35,.35,.35}
\definecolor{vvblue}{rgb}{.14,.21,.868}%
\renewcommand*{\irrulename}[1]{\text{\textcolor{vblue}{\upshape#1}}}%

\newcommand{\R}{\mathbb{R}}
\newcommand{\N}{\mathbb{N}}
\newcommand{\Q}{\mathbb{Q}}

\newcommand{\lra}{\textrm{LRA}}

\renewcommand{\phi}{\varphi}

\newcommand{\eps}{\varepsilon}
\newcommand{\eval}[1]{\left\llbracket #1 \right\rrbracket}
\newcommand{\norm}[1]{\left\lVert#1\right\rVert}
\newcommand{\bebecomes}{\mathrel{::=}}
\newcommand{\alternative}{~|~}

\newcommand{\folr}{\textsf{FOL}_{\R}}

\newcommand{\lip}{C^{0, 1}}
\newcommand{\norml}[1]{\norm{#1}_{L}}
\newcommand{\normi}[1]{\norm{#1}_{\infty}}
\newcommand{\normc}[1]{\norm{#1}_{C^{0, 1}}}

\newcommand{\Acc}{\mathsf{Acc}}
\newcommand{\Rad}{\mathsf{Rad}}
\newcommand{\End}{\mathsf{End}}

\newcommand{\flow}[1]{x'=Ax,t'=1{\&}t\leq #1}

\renewcommand*{\lie}[3][]
{\mathcal{L}_{#2}^{\ifthenelse{\equal{#1}{}}{}{^{\left(#1\right)}}}(#3)}

\renewcommand*{\lied}[3][]{\overset{\bm .}{#3}\ifthenelse{\equal{#1}{}}{}{{}^{(#1)}}}

\renewcommand{\Dostar}[1]{\ifthenelse{\equal{#1}{}}{(*)}{-(*)}}

\newsavebox{\Rval}%
\sbox{\Rval}{$\scriptstyle\mathbb{R}$}
\newsavebox{\backiterateb}%
\sbox{\backiterateb}{$\scriptstyle\overleftarrow{\dibox{{}^*}}$}

\newcommand{\fvarA}{\phi}
\newcommand{\fvarB}{\psi}

\newcommand{\rfvar}{P}

\newcommand{\D}[1]{#1'}

\title{Linear Axiomatization of Linear Differential Equations}

\author{Long Qian}
{Carnegie Mellon University, Pittsburgh, PA, USA}
{longq@andrew.cmu.edu}
{https://orcid.org/0000-0003-1567-3948}
{}

\author{Andr\'e Platzer}
{Karlsruhe Institute of Technology, Karlsruhe, Germany}
{platzer@kit.edu}
{https://orcid.org/0000-0001-7238-5710}
{}
\Copyright{Long Qian and Andr\'e Platzer}
\authorrunning{L. Qian and A. Platzer}
\hideLIPIcs



\acknowledgements{This work was funded by an Alexander von Humboldt Professorship.}

\begin{document}
\nolinenumbers

\maketitle

\begin{abstract}
This article presents an axiomatization of \emph{linear} differential equations that is \emph{complete} for topologically open properties with compact initial conditions and time horizons by proving them deductively via a sound reduction to formulas of linear real arithmetic. 
Indeed, linearity is preserved: \emph{linear} differential invariants are complete for open properties of \emph{linear} differential equations and reduce with a \emph{linear} number of Lie derivatives (per atom) to \emph{linear} real arithmetic.
In particular, this article shows that the linear fragment of differential dynamic logic (\dL), a modal extension of standard first-order logic over the reals satisfies these completeness properties.
By establishing completeness within decidable linear real arithmetic, the need for proving sophisticated nonlinear formulas of real arithmetic is circumvented.
Furthermore, this article establishes generalizations to the nonlinear fragment of \dL that enable the practical certification of reachability computations, validating the output of external computations in \dL with machine-checkable proofs.
\end{abstract}

\section{Introduction}
Ordinary differential equations (ODEs) are well-known for their ubiquitous applicability across many disciplines, highlighting their importance. Classically, ODEs are often studied quantitatively by analyzing numerical approximations to their solutions, leading to the development of many sophisticated and deep algorithms to compute such approximations in the field of numerical analysis. Alternatively, a seemingly orthogonal approach is to analyze ODEs qualitatively \cite{Poincare81}, where properties of their solutions are deduced directly from the differential equations. The differences in these approaches are particularly highlighted in the safety verification of cyber-physical/hybrid dynamical systems \cite{Alur_2015}. Quantitative approaches are often viewed to be more scalable for verifying numerical properties of a system \cite{DBLP:conf/cav/ChenAS13, DBLP:conf/cpsweek/Althoff15, DBLP:conf/itp/0001T24}, but the intrinsic complexity in the algorithms adopted makes it difficult to justify their correctness. In contrast, qualitative approaches via logical axiomatizations \cite{DBLP:journals/jar/Platzer08} are symbolic, with correctness being easily justified by having a small set of soundness-critical axioms, yielding machine-checkable proofs of correctness \cite{Fulton_Mitsch_Quesel_Volp_Platzer_2015, DBLP:conf/cpp/BohrerRVVP17}. 

Recent work has shown that, surprisingly, quantitative and qualitative approaches to ODEs can be integrated harmoniously \cite{compact_IVP} in \emph{differential dynamic logic} (\dL). In particular, it was established that, given any (untrusted) approximation to the solution of an ODE, if the approximation is truly accurate, then a corresponding deductive proof can be computed in the axiomatization of $\dL$ certifying its accuracy --- a completeness result for approximations. However, the complete axiomatization and algorithm suffer from the following drawbacks:
\begin{enumerate}
    \item Given a (potentially untrusted) numerical approximation to the solution of an ODE, its accuracy is certified by deductively proving differential invariants \cite{DBLP:journals/jar/Platzer17}. Although \dL is complete for differential invariants and soundly reduces them to formulas of real arithmetic \cite{DBLP:journals/jacm/PlatzerT20}, the proof for each differential invariant instance requires taking $N$ Lie derivatives, where $N$ is some finite number dependent on the ODE and the formula of invariance, and is doubly exponential in the number of variables in general \cite{DBLP:journals/jacm/PlatzerT20, differential_rank_bound}. 
    \item As \dL is a modal extension of the standard first-order logic of real arithmetic ($\folr$), deductive proofs (e.g., of differential invariants) often produce proof obligations of $\folr$ formulas, at which point Tarski's completeness theorem \cite{Tarski_1948} is utilized to conclude that the $\folr$ formula is provable (iff it is valid) using decision procedures for general real arithmetic, which also suffer from doubly exponential complexity \cite{DBLP:journals/jsc/DavenportH88}.  
    \item Due to the need of using general decision procedures of $\folr$ and the fact that most such algorithms are not proof-producing, this loses a layer of reliability in the process and only yields machine-checkable proofs of correctness modulo $\folr$ \cite{DBLP:conf/cpp/KosaianTP23}. 
\end{enumerate}

The main contribution of this article is to show that, surprisingly, by specializing to the linear case, linearity is preserved - (open) properties (on finite time horizons) definable in \emph{linear} real arithmetic ($\lra$) of \emph{linear} ODEs are provable in the \emph{linear} fragment of \dL using \emph{linear} differential invariants requiring at most a \emph{linear number} of Lie derivatives (per atom, in the dimension of the ODE), deductively reducing proofs to the provability of $\lra$ formulas. In particular, this resolves the drawbacks pointed out above, and the existence of efficient, proof-producing decision procedures for $\lra$ \cite{DBLP:conf/cade/BoutonODF09} further enables a complete proof-producing axiomatization down to linear real arithmetic. Note that the completeness results of this article do not follow from earlier results \cite{compact_IVP} by simply restricting to the linear fragment, which require the use of nonlinear differential invariants even in the linear case. Despite linear ODEs being a strict subset of the general class of polynomial ODEs, their properties in relation to computability/provability are far from trivial. Earlier works \cite{DBLP:journals/tcs/BellDJB10, DBLP:conf/icalp/ChonevOW16} already demonstrated that even (seemingly) basic questions such as decidability of hyperplane reachability are open and related to deep open problems in number theory \cite{DBLP:journals/tcs/BellDJB10, DBLP:conf/icalp/ChonevOW16}. 

While the results of this article are theoretical in nature, they have practical implications, particularly in the case of certifying over-approximations of reachable sets for linear ODEs. In the safety verification of cyber-physical systems, a frequently encountered problem is to compute an over-approximation of all states reachable from a (compact) set of initial conditions $C \subseteq\R^n$ following the linear ODE $x' = Ax$ on some time interval $[0, T]$, and the system is safe if this over-approximation is disjoint from a given set of unsafe states. In safety-critical applications, it is important to have a trustworthy verification that the computed over-approximations are truly correct beyond basic mathematical rigor of the underlying algorithms used, ideally in the form of a formal, machine-checkable proof. Previous works aiming at bridging this gap focused on formalizing particular reachability algorithms within theorem provers \cite{DBLP:journals/jar/Immler18, DBLP:conf/itp/0001T24} such that the over-approximations can be trusted. 

An alternative viewpoint offers a different approach - one can consider the problem of safety verification as ``satisfiability (SAT) solving modulo ODEs \cite{DBLP:conf/atva/EggersFH08,Gao_Kong_Clarke_2013-dreal}''. A system is safe if and only if the formula ``$x_0 \in C \land t \in [0, T] \land \phi(x_0, t) \in \text{UNSAFE}$'' is unsatisfiable (UNSAT), where $\phi(x_0, t)$ denotes the flow of the ODE at time $t$ starting at $x_0$.
Even in the field of SAT solving for propositional logic, SAT solvers are complicated and not typically verified themselves.
Instead, SAT solvers produce \emph{proofs of unsatisfiability}, which can then be checked with a formally verified proof checker. 
Tools for safety verification of ODEs are even more difficult, and yet no proof checking approach exists, because, unlike for SAT solving, it is not even clear what underlying axiomatic system to use for certifying proofs with ODEs. The results of this article precisely demonstrate that for linear ODEs, the linear fragment of \dL is fully capable of expressing such proofs, which can then be checked against any prover for \dL, such as KeYmaera X \cite{Fulton_Mitsch_Quesel_Volp_Platzer_2015} or its formalizations in Isabelle/HOL, Rocq \cite{DBLP:conf/cpp/BohrerRVVP17} or PVS \cite{DBLP:journals/corr/abs-2404-15214}. 

Finally, this article also establishes practical results beyond general completeness, proving that a concrete reachability algorithm for (homogeneous) linear ODEs \cite{DBLP:conf/hybrid/Girard05} can be carried out in (nonlinear) $\dL$ (\rref{thm:step_approx}), yielding proofs of accuracy certified by KeYmaera X (\rref{fig:2d-flowpipe}). A certification procedure that can certify (untrusted) over-approximations of reachable sets is also established (\rref{cor:step_cert}). While these results (\rref{thm:step_approx}, \rref{cor:step_cert}) require extending to the nonlinear fragment, the use of nonlinear real arithmetic is limited to establishing a local Taylor approximation, and only requires interval bounds. Hence, quantifier elimination of general nonlinear real arithmetic is not required, and neither are general completeness results of differential invariants.

\begin{figure}
    \centering
    \includegraphics[
      width=0.6\linewidth
    ]{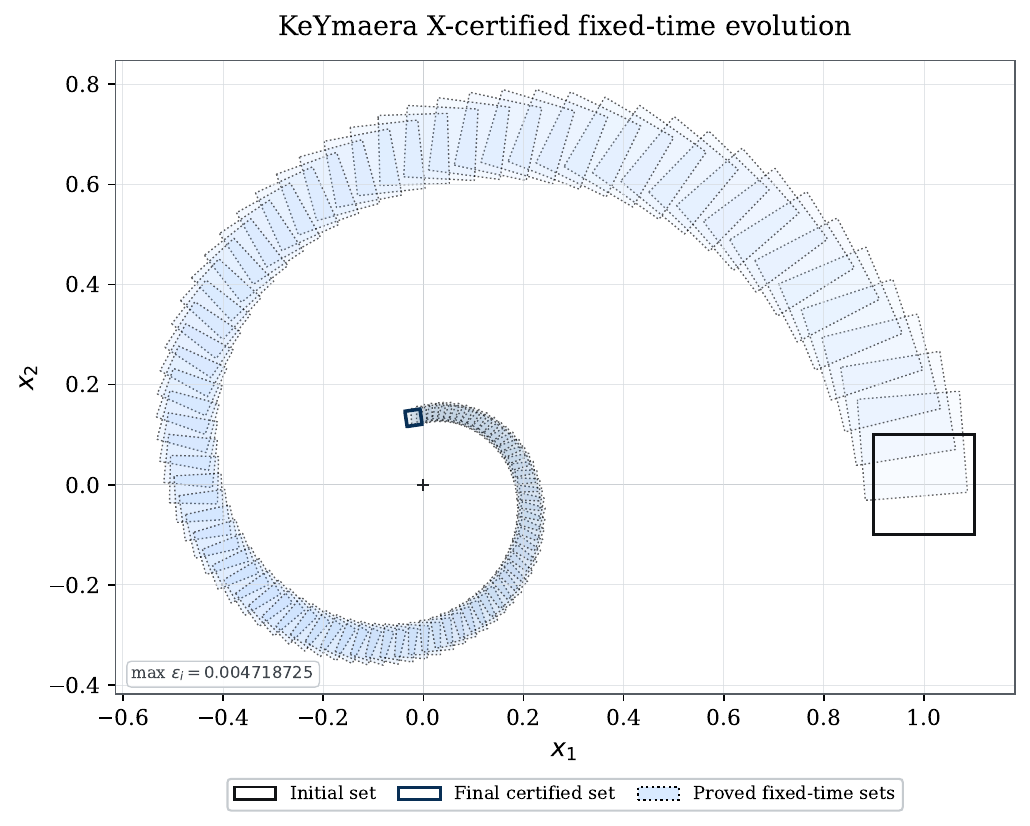}
    \caption{Certified evolution of $(x_1, x_2)$ over $t \in [0, 2]$ with provable error $\approx 4 \times 10^{-3}$ constructed via Theorem \ref{thm:step_approx}. ODE taken from reachability example \cite{DBLP:conf/hybrid/Girard05}.}
    \label{fig:2d-flowpipe}
  \end{figure}

\section{Differential Dynamic Logic}
This section provides the necessary background on differential dynamic logic (\dL) \cite{DBLP:journals/jar/Platzer08,DBLP:journals/jar/Platzer17} that is needed for the article. On a high level, \dL is a modal extension of the standard first-order logic over the reals, where modalities are annotated with (polynomial) ODEs. This article focuses on the linear, continuous fragment of \dL.

\subsection{Syntax and Semantics}
Terms in the linear fragment of \dL are built from the following grammar, where $x$ denotes a variable, and $c \in \Q$ a rational constant. Equivalently, terms are (multivariate) affine polynomials of the variables with rational coefficients.
\[p, q\bebecomes x\alternative c\alternative p + q\alternative c\cdot q\]
Similarly, $\lra$ refers to the linear fragment of the theory of real arithmetic, only allowing for constant (rational) multiplications. For full (nonlinear) $\dL$, the product of terms $p \cdot q$ is also allowed. In this article, vectorial variables are denoted as $x = (x_1, \dots, x_n)$ for presentation.
\dL formulas have the following grammar, where ${\preceq} \in \{<, \leq\}$ is a comparison relation, $A \in \Q^{n \times n}$ is a rational matrix defining the ($n$-dimensional) linear ODE $x' = Ax$ and $T \in \Q^+$. 
\begin{align*}
    \phi, \psi &\bebecomes p {\preceq} q \alternative \phi \land \psi \alternative \phi \lor \psi \alternative \neg \phi \alternative \forall x \phi \alternative \exists x \phi \alternative 
\ddiamond{\alpha}{\psi} \alternative \dbox{\alpha}{\psi}\\
\alpha &\bebecomes \cdots \alternative x' = Ax, t' = 1 \,\&\, t \leq T
\end{align*}
This article only concerns the case where $\alpha$ is a linear ODE, but \dL allows $\alpha$ to be of the more general class of hybrid programs, which are not needed here. The semantics of non-modal \dL formulas are standard and agree with their interpretations in $\folr$. The modal formulas $\dbox{\alpha}{\phi}, \ddiamond{\alpha}{\phi}$ express \emph{safety} and \emph{liveness} properties, respectively, and can be thought of as universal/existential quantifiers over time in the flow of $\alpha$. Concretely, let $\phi(x_0, t) = e^{At}x_0$ denote the flow of $x' = Ax$ with initial condition $x(0) = x_0$, then \(\dbox{x' = Ax, t' = 1 \& t \leq T}{\psi}\) expresses the property that $\psi(x(t))$ holds for all $t \in [0, T]$ and $x(t) = \phi(x_0, t)$, where $x_0$ is the value of $x$ in the current state. Likewise, \(\ddiamond{\alpha}{\psi}\) expresses the property that there exists some $t \in [0, T]$ such that $\psi(x(t))$ holds. The constraint $t \leq T$ is referred to as the \emph{domain constraint} and can more generally be any formula of real arithmetic. For a \dL formula $\phi$ in $n$-variables, $\eval{\phi} \subseteq \R^n$ denotes the set of all states that satisfy $\phi$.  A \dL formula $\phi$ is valid iff it is true in all states, i.e., \(\eval{\phi} = \R^n\). 

\subsection{Proof Calculus}
The proof calculus of \dL is a standard, classical sequent calculus with the standard proof rules for first-order reasoning. A sequent has the form $\Gamma \vdash \phi$, equivalent in semantics to the implication $\left(\bigwedge_{\psi \in \Gamma} \psi\right) \rightarrow \phi$. An axiom (schema) is said to be \emph{sound} if and only if all of its instances are valid, and a proof rule is \emph{sound} if the validity of its premises imply the validity of its conclusion. Axioms/proof rules are \emph{derivable} if they can be deductively proven in the aforementioned sequent calculus. The soundness of \dL \cite{DBLP:journals/jar/Platzer08} implies that all derivable axioms/proof rules are sound and all derivable formulas valid.

This article concerns the continuous fragment of $\dL$ under the standard axiomatization \cite{compact_IVP}. An important property of \dL is that it is complete for all \emph{differential invariants}, that is, if a $\folr$ formula $I$ characterizes a positive invariant of the ODE $x' = f(x)$, then it is provably invariant \cite{DBLP:journals/jacm/PlatzerT20, DBLP:conf/emsoft/LiuZZ11}. Furthermore, the proof of its invariance is soundly reduced to real arithmetic formulas concerning the Lie derivatives of $I$ \cite{DBLP:conf/emsoft/LiuZZ11,DBLP:journals/jacm/PlatzerT20}, hence linearity is preserved. The following theorem concerns the specialization of this completeness result to the linear case, which this article relies upon. It was implicit in earlier work \cite{DBLP:journals/jsc/GhorbalS22} and is proved here for completeness. 

\begin{theorem}[Completeness of Linear Differential Invariants \cite{DBLP:journals/jsc/GhorbalS22}]
    \label{thm:complete-invariants}
    \dL is complete for linear differential invariants. For all $A \in \Q^{n \times n}$, $\lra$ formula $I$, if the \dL formula 
    \[I \rightarrow \dbox{x' = Ax \& Q}{I}\]
    is valid, then one can effectively find a proof of it in the linear fragment of \dL by taking at most $n$ Lie derivatives of $I$. $Q$ here either denotes $t \leq T$ for some $T \in \Q^+$, in which case the vectorial ODE $x' = Ax$ includes the clock variable $t' = 1$, or $Q$ is the trivially true formula $\top$ (hence reduces to $x' = Ax$). 
\end{theorem}
The provability parts of Theorem \ref{thm:complete-invariants} follow directly from the proof of completeness of semi-algebraic invariants \cite[Theorem~6.9]{DBLP:journals/jacm/PlatzerT20}. The following provides the arguments needed for the linearity claims for completeness. 
\begin{proof}
    It suffices to follow the proof of general completeness \cite[Theorem~6.9]{DBLP:journals/jacm/PlatzerT20} and verify that linearity is preserved. First note that as $I$ is $\lra$ definable, we may assume it is of the form $I(x) \equiv \alpha^T x \preceq 0$ for some $\alpha \in \Q^n$ ($\alpha^Tx$ is then the dot product) and ${\preceq} \in \{<, \leq\}$, as a general Boolean combination of these relations can be handled similarly \footnote{As local progression formulas justifying invariance distribute across logical connectives \cite[Lemma~6.5]{DBLP:journals/jacm/PlatzerT20}.}. On a high level, the proof of completeness for differential invariants proceeds as follows \cite{DBLP:conf/emsoft/LiuZZ11}:
    \begin{itemize}
        \item Compute the chain of Lie derivatives of $\alpha^Tx$ with respect to the vector field $x' = Ax$:
        \[\alpha^Tx, (\alpha^Tx)' = \alpha^TAx, (\alpha^Tx)'' = \alpha^TA^2x, \dots\]
        As each term is a polynomial in $x$ and $\R[x]$ is Noetherian, there must be some finite index $N$ such that the ideal generated by these polynomials stabilizes at the $N$-th derivative.
        \item A $\lra$ formula is then constructed using the first $N$ Lie derivatives, such that the validity of this formula is provably equivalent to the invariance of $I$.
    \end{itemize}
    It can first be observed that linearity is indeed preserved, as Lie derivatives of $\alpha^Tx$ are always linear in $x$ (recall that $A$ is not symbolic but rational). Secondly, note that by Cayley-Hamilton, $A^n$ is a linear combination of $\{I, A, \dots, A^{n - 1}\}$, therefore the sequence of Lie derivatives stabilizes at $\alpha^TA^nx$, completing the proof. 
\end{proof}
A complete record of (the fragment of) $\dL$'s axiomatization used for this article is provided in \rref{app:axioms} for reference. 

\section{Linear Axiomatization}
\label{section: linear axiomatization}
This section establishes the completeness of linear ODEs in the linear fragment of \dL, in particular, any sequence of \emph{admissible} approximations converging to the solution of the ODE converges provably. In the rest of this section, we work in the linear fragment of $\dL$. Let $x' = Ax$ be a linear ODE with $A \in \Q^{n \times n}$, $C(x) \in \lra$ be a set of initial states definable in linear arithmetic that is topologically compact (i.e. $\eval{C} \subseteq \R^n$ is compact), and $[0, T]$ a positive time interval with $T \in \Q^+$. We write $\phi(x_0, t) \coloneqq e^{At}x_0$ to denote the flow of the ODE with initial condition $x(0) = x_0$ because the general solution for linear ODEs has this form with the matrix exponential $e^{At}$. The following class of approximations can be viewed as the linearized version of local definable approximations \cite[Definition~4.8]{compact_IVP}, representing numerical approximations. 

\begin{definition}[Admissible Approximations]
    \label{def:admissible-approx}
    Let $(\Phi_k(x_0, t))_k \in C^0(\eval{C} \times [0, T], \R^n)$ be a sequence of continuous functions. This sequence is said to be an \emph{admissible approximation} if:
    \begin{itemize}
        \item For every $k$, the function $\Phi_k(x_0, t)$ is definable in $\lra$. 
        \item The sequence $\Phi_k(x_0, t)$ converges to $\phi(x_0, t)$ in $\lip([0, T], \R^n)$ uniformly in $x_0$. That is, the following sequence converges to $0$
        \[\sup_{x_0 \in \eval{C}} \normc{\Phi_k(x_0, t) - \phi(x_0, t)} \xlongrightarrow{k} 0\]
    \end{itemize}
    where $C^{0, 1}([0, T])$ denotes the space of Lipschitz continuous functions endowed with the norm $\normc{f} = \normi{f} + \norml{f}$, $\normi{f} = \sup_{t \in [0, T]} \abs{f(t)}$ and $\norml{f} = \sup_{x \neq y \in [0, T]} \frac{\normi{f(x) - f(y)}}{\normi{x - y}}$ is the least Lipschitz constant. 
\end{definition}

Intuitively, such admissible approximations represent the numerical approximations to the true flow $\phi(x_0, t)$. The only requirements are that they are $\lra$ definable (e.g. obtained by interpolation) and they converge to the true solution. Importantly, there are no \emph{provability} requirements. The main theorem of this section states that the convergence of such approximations can be deductively proven in \dL.

\begin{theorem}[Linear Completeness of Convergence]
    \label{thm:complete-convergence}
    Let $(\Phi_k(x_0, t))_k$ be a sequence of admissible approximations and $\eps \in \Q^+$ be arbitrary. The following formula is provable in the linear fragment of $\dL$ for all sufficiently large $k$
    \[C(x) \land x = x_0 \land t = 0 \rightarrow \dbox{x' = Ax, t' = 1 \& t \leq T}{\normi{x - \Phi_k(x_0, t)} < \eps}\]
    Note that for vectors $\normi{\cdot}$ denotes the infinity-norm (i.e. maximum of absolute values of coordinates), thus the formula $\normi{x - \Phi_k(x_0, t)} < \eps$ belongs to the linear fragment of \dL (recall that $\Phi_k$ is required to be $\lra$ definable). 
\end{theorem}

Theorem \ref{thm:complete-convergence} establishes that given any sequence of admissible approximations, corresponding proofs of $C^0$ convergence can be generated, entirely within \dL. It is important to observe that this theorem only places \emph{semantic} requirements ($\lip$ convergence) on the approximations, and guarantees a \emph{syntactic} proof of accuracy. The following lemma is a key step in the proof of Theorem \ref{thm:complete-convergence}, reducing the problem down to \emph{linear} differential invariants. 

\begin{lemma}[Linear Differential Invariant]
    \label{lem:linear-invariant}
    Let $(\Phi_k(x_0, t))_k$ be a sequence of admissible approximations, $\normi{A} \leq L \in \Q^+$ an upper bound on the (infinity) norm of $A$ and $h \in \Q^+$ an arbitrary constant. Then the following is a valid (linear) differential invariant for all sufficiently large $k$. 
    \[
        \psi_k(x_0, x, g, v, t) \rightarrow
         \dbox{x' = Ax, g' = Lg, v' = Lv + g, t' = 1 \& t \leq T} \psi_k(x_0, x, g, v, t)
    \]
    where the invariant $\psi_k(x_0, x, g, v, t)$ is defined via 
    \begin{align*}
        \eps(g, v) &\equiv h(2g + v - 1)\\
        \psi_k(x_0, x, g, v, t) &\equiv 
        t \geq 0 \land g \geq 1 \land v \geq 0 \land C(x_0) \land \normi{\Phi_k(x_0, t) - x} \leq \eps(g, v)
    \end{align*}
\end{lemma}

\rref{lem:linear-invariant} can intuitively be understood as a characterization of the ``continuous dependence on initial conditions'' property of ODE flows. In particular, the \dL formula $\psi_k(x_0, x, g, v, t)$ captures the notion that $\Phi_k$ is ``$\eps(g, v)$-close'' to the true flow $\phi(x_0, t)$. While this may appear similar in structure to an analogous result in the nonlinear case \cite[Lemma~4.11]{compact_IVP}, an important difference is that the cited result uses nonlinear invariants (even when $\Phi_k$ is $\lra$ definable), and is therefore not expressible in the linear fragment of $\dL$ and cannot establish such completeness properties. 

\begin{proof}[Proof of Lemma \ref{lem:linear-invariant}]
The validity of this differential invariant can be established directly.
Suppose that $\psi_k(x_0, x, g, v, t)$ is satisfied at some initial state $(x_0, x, g, v, t) = (y_0, y, g_0, v_0, t_0)$, thus the following conditions hold:
\begin{align*}
    \psi_k(y_0, y, g_0, v_0, t_0) &~\equiv~
    t_0 \geq 0 \land g_0 \geq 1 \land v_0 \geq 0 \land C(y_0) \land \normi{y - \Phi_k(y_0, t_0)} \leq \eps(g_0, v_0)
\end{align*}
As the domain constraint requires $t \leq T$, assume without loss of generality that $t_0 \leq T$ holds as well. Recall that $\phi(x_0, t) = e^{At}x_0$ denotes the flow of $x' = Ax$. Let $\phi_g(g_0, t) = e^{Lt}g_0$ denote the flow of $g' = Lg$, and $\phi_v(v_0, t) = (v_0 + g_0t)e^{Lt}$ the flow of $v' = g + Lv$. By definition of flows, at time $t_1 \in [t_0, T]$ the variables evolve to $(y, g_0, v_0) \mapsto (\phi(y, t_1 - t_0), \phi_g(g_0, t_1 - t_0), \phi_v(v_0, t_1 - t_0))$. Due to the monotonicity of $t' = 1$ and $y_0$ being constant across the flow, it suffices to show that the following inequality holds for all $t_1 \in [t_0, T]$
\[\normi{\Phi_k(y_0, t_1) - \phi(y, t_1 - t_0)} \leq \eps(\phi_g(g_0, t_1 - t_0), \phi_v(v_0, t_1 - t_0))\]
which implies the validity of the differential invariant. The remainder of the proof then focuses on bounding {$\normi{\Phi_k(y_0, t_1) - \phi(y, t_1 - t_0)}$}. Let $R_k(x_0, t) \coloneqq \Phi_k(x_0, t) - \phi(x_0, t)$ denote the difference between the approximation and the true flow. Direct manipulations via the semigroup property of flows then give
\begin{align*}
    &\Phi_k(y_0, t_1) - \phi(y, t_1 - t_0) = \\
    &e^{A(t_1 - t_0)}(\Phi_k(y_0, t_0) - y) + (R_k(y_0, t_1) - R_k(y_0, t_0)) + (I - e^{A(t_1 - t_0)})R_k(y_0, t_0)
\end{align*}

Recall that $\Phi_k(x_0, t)$ converges to $\phi(x_0, t)$ in the $\lip$ norm, so we may take $k$ to be large enough such that $\normi{R_k(y_0, t_0)} \leq h$ and $\normi{R_k(y_0, t_0) - R_k(y_0, t_1)} \leq h \abs{t_1 - t_0}$, substituting this and $\normi{\Phi_k(y_0, t_0) - y} \leq \eps(g_0, v_0)$ into the bound above yields
\begin{align*}
    &\normi{\Phi_k(y_0, t_1) - \phi(y, t_1 - t_0)}\\
    \leq {}& \eps(g_0, v_0)e^{\normi{A}(t_1 - t_0)} + h\abs{t_1 - t_0} + h(e^{\normi{A}(t_1 - t_0)} - 1)\\
    = {}& (\eps(g_0, v_0) + h)e^{\normi{A}(t_1 - t_0)} + h\abs{t_1 - t_0} - h\\
    \leq {}& h(2g_0 + v_0)e^{L(t_1 - t_0)} + h\abs{t_1 - t_0} - h\\
    \leq {}& h(2g_0e^{L(t_1 - t_0)} + (v_0 + g_0(t_1 - t_0))e^{L(t_1 - t_0)} - 1) + h\abs{t_1 - t_0} - h(t_1 - t_0)\\
    = {}& \eps(\phi_g(g_0, t_1 - t_0), \phi_v(v_0, t_1 - t_0))
\end{align*}
which is exactly what we wanted, establishing the validity of the invariant. It is worth pointing out that from the computations above, given $h \in \Q^+$ as a parameter, choosing (any) $k$ large enough such that $\sup_{x_0 \in \eval{C}} \normc{\Phi_k(x_0, t) - \phi(x_0, t)} \leq h$ will be sufficient to witness validity of $\psi_k$. 
\end{proof}
A consequence of Lemma \ref{lem:linear-invariant} is that if the bound $\normi{\Phi_k(x_0, 0) - x_0} \leq \eps(1, 0) = h$ is satisfied, then this bound holds on all of $[0, T]$ following the respective (linear) ODEs. Thus, this implies that $\Phi_k(x_0, t)$ approximates the true flow $\phi(x_0, t)$ on all of $\eval{C} \times [0, T]$ with accuracy 
\[\eps(\phi_g(1, t), \phi_v(0, t)) = h(2e^{Lt} + te^{Lt} - 1) \leq h(2 + T)e^{LT}\]
which converges to $0$ as $h \to 0$ by taking $k \to \infty$. Consequently, Theorem \ref{thm:complete-convergence} can be established by taking $k$ sufficiently large such that $h(2 + T)e^{LT} < \eps$. A challenge in this argument is that the functions $e^{Lt}$ (encoded by $g' = Lg$) and $te^{Lt}$ (encoded by $v' = g + Lv$) are \emph{not} directly expressible in the linear arithmetic fragment of $\dL$, nor in polynomial arithmetic, so one has to justify that arbitrarily tight upper bounds can be provably established.
\begin{lemma}[Exponential Bounds]
    \label{lem:exp-bound}
    Let $L, M_1, M_2, T \in \Q^+$ be constants such that the inequalities $e^{LT} < M_1, Te^{LT} < M_2$ hold, then the following formula is provable in the linear fragment of \dL
    \[
        g = 1 \land v = 0 \land t = 0 \rightarrow
        \dbox{g' = Lg, v' = Lv + g, t' = 1 \& t \leq T}{(g < M_1 \land v < M_2)}
    \]
\end{lemma}
\begin{proof}
    The bounds $M_1, M_2$ are proved similarly, with $g < M_1$ proven first. We first claim that it is possible to construct some (continuous) piecewise linear function $f(t)$ satisfying:
    \begin{enumerate}
        \item $f(T) < M_1$.
        \item $f(0) \geq 1$. 
        \item On each linear segment $[t_i, t_{i + 1})$ of $f(t)$, the linear function satisfies $f'(t) \geq Lf(t)$ for all $t \in [t_i, t_{i + 1})$. 
    \end{enumerate}
    Indeed, such a function $f(t)$ can be constructed by iteratively applying linear interpolation. Let $N \in \N$ be a sufficiently large integer denoting the number of time steps that is to be decided later, this gives a step size of $h = \frac{T}{N}$ decomposing the interval $[0, T]$ into $\{[0, h], [h, 2h], \dots, [(N - 1)h, Nh = T]\}$. Define a monotone sequence of values $\{y_k\}_{0 \leq k \leq N}$ inductively as follows:
    \begin{itemize}
        \item $y_0 = 1$.
        \item $y_{k + 1} = \frac{y_k}{1 - Lh}$ (take $N$ large enough such that $Lh < 1$).
    \end{itemize}
    Define $f(t)$ on the interval $[kh, (k + 1)h)$ by linear interpolation
    \[f(t) = y_{k} + \frac{y_{k + 1} - y_k}{h}(t - kh)\]
    It suffices to verify that this satisfies the desired properties, clearly (2) is satisfied as $f(0) = y_0 = 1 \geq 1$. For property (3), note that on a interval $[kh, (k + 1)h]$ the derivative of $f(t)$ is $f'(t) = \frac{y_{k + 1} - y_k}{h}$, as $y_{k + 1}(1 - Lh) = y_{k}$, this implies
    \[f'(t) = \frac{y_{k + 1} - y_k}{h} = \frac{Lhy_{k + 1}}h = Ly_{k + 1} \geq Lf(t)\]
    where the final inequality follows from $y_{k} \leq y_{k + 1}$. Thus, it remains to establish that by taking $N$ sufficiently large $y_{N} < M_1$ holds. Indeed, the construction of $y_{N}$ gives
    \[\lim_{N \to \infty} y_N = \lim_{N \to \infty} \left(1 - \frac{TL}{N}\right)^{-N} = e^{LT} < M_1\]
    therefore all sufficiently large $N$ will witness $f(T) = y_{N} < M_1$. To complete the proof of the lemma, notice that the following holds:
    \begin{itemize}
        \item To prove $g < M_1$ (in \dL), it suffices to prove $g \leq f(t)$ holds on $[0, T]$ in \dL, i.e. prove the following 
        \[g = 1 \land t = 0 \rightarrow \dbox{g' = Lg, t' =1 \& t \leq T}{f(t) - g \geq 0}\]
        \item On each interval of $f(t)$, computing the derivative of the difference gives 
        \[(f(t) - g)' = f'(t) - Lg \geq Lf(t) - Lg = L(f(t) - g)\]
        That is, on each interval the ($\lra$ definable) term $f(t) - g$ forms a Darboux polynomial (with cofactor $L \in \Q^+$, so linearity is preserved) of the system \cite[Section~3]{DBLP:journals/jacm/PlatzerT20} and $f(t) - g \geq 0$ therefore holds. The overall inequality $f(t) - g \geq 0$ on $[0, T]$ is therefore a valid differential invariant
        \[f(t) - g \geq 0 \rightarrow \dbox{g' = Lg, t' = 1 \& t \leq T} f(t) - g \geq 0\]
        and hence provable via \rref{thm:complete-invariants}, completing the proof for the bound of $g < M_1$. 
    \end{itemize}
    The proof of $v < M_2$ proceeds similarly, first defining a monotone sequence $\{z_k\}_{0 \leq k \leq N}$ via
    \begin{itemize}
        \item $z_0 = 0$
        \item $z_{k + 1} = \frac{z_k + hy_{k + 1}}{1 - hL}$
    \end{itemize}
    Now construct a piecewise linear function $u(t)$ such that $u(t) = z_k + \frac{z_{k + 1} - z_k}{h}(t - kh)$ on $[kh, (k + 1)h]$. A similar analysis of the derivative of $u(t)$ on $[kh, (k + 1)h]$ gives
    \[u'(t) = \frac{z_{k + 1} - z_k}{h} = \frac{z_{k + 1}hL + hy_{k + 1}}{h} = Lz_{k + 1} + y_{k + 1} \geq Lu(t) + f(t)\]
    Thus, this gives (recall that $f(t) - g \geq 0$ has been proven)
    \[(u(t) - v)' \geq Lu(t) + f(t) - g - Lv \geq L(u(t) - v)\]
    Thus $u(t) - v \geq 0$ is also a valid (linear) differential invariant and therefore provable, the bound $v < M_2$ follows by establishing that $u(T) \to Te^{LT}$ as $N \to \infty$. Indeed, it is easy to verify by induction that one has the identity $z_{k} = \frac{kh}{(1 - hL)^{k + 1}}$, taking the limit as $N \to \infty$ of $z_N$ then yields
    \[\lim_{N \to \infty} z_N = \lim_{N \to \infty} \frac{Nh}{(1 - hL)^{N + 1}} = Te^{LT}\]
    which is what we wanted, completing the proof. 
\end{proof}

With Lemma \ref{lem:exp-bound}, the proof of Theorem \ref{thm:complete-convergence} can now be completed. 

\begin{proof}[Proof of Theorem \ref{thm:complete-convergence}]
    Let $\eps \in \Q^+$ be the desired accuracy, and let $k$ be large enough such that the following inequality holds
    \[ \left(\sup_{x_0 \in \eval{C}} \normc{\Phi_k(x_0, t) - \phi(x_0, t)} \right)(2 + T)e^{\normi{A}T} < \eps\]
    Let $L \coloneqq \normi{A}$ and $h, M_1, M_2 \in \Q^+$ be constants such that the following conditions hold
    \begin{enumerate}
        \item $\sup_{x_0 \in \eval{C}} \normc{\Phi_k(x_0, t) - \phi(x_0, t)} < h$
        \item $e^{\normi{A}T} < M_1$
        \item $Te^{\normi{A}T} < M_2$
        \item $h(2M_1 + M_2) < \eps$
    \end{enumerate}
    To proceed, first apply Lemma \ref{lem:linear-invariant} and conclude (by condition (1)) that the following is a valid (and therefore provable by \rref{thm:complete-invariants}) differential invariant
    \begin{align*}
        &\psi_k(x_0, x, g, v, t) \rightarrow \dbox{x' = Ax, g' = Lg, v' = Lv + g, t' = 1 \& t \leq T} \psi_k(x_0, x, g, v, t)
    \end{align*}
    Further note that with initial conditions $g = 1, v = 0, t = 0, x = x_0$, the formula $\psi_k(x_0, x_0, 1, 0, 0)$ is satisfied (recall that $\eps(1, 0) = h(2 + 0 - 1) = h$), hence following is derivable 
    {\small
    \begin{align*}
        &C(x_0) \land x = x_0 \land t = 0 \land g = 1 \land v = 0 \rightarrow\\
        &\dbox{x' = Ax, g' = Lg, v' = Lv + g, t' = 1 \& t \leq T}{\normi{\Phi_k(x_0, t) - x} \leq h(2g + v - 1)}
    \end{align*}
    }
    Invoking Lemma \ref{lem:exp-bound} then proves the following formula
    \begin{align*}
       & t = 0 \land g = 1 \land v = 0 \rightarrow \dbox{g' = Lg, v' = Lv + g, t' = 1 \& t \leq T}{h(2g + v - 1) \leq h(2M_1 + M_2)}
    \end{align*}
    Combining the two formulas above then proves 
    {\small
    \begin{align*}
        &C(x_0) \land x = x_0 \land t = 0 \land g = 1 \land v = 0 \rightarrow\\
        &\dbox{x' = Ax, g' = Lg, v' = Lv + g, t' = 1 \& t \leq T}{\normi{\Phi_k(x_0, t) - x} \leq h(2M_1 + M_2)}
    \end{align*}
    }
    since $h(2M_1 + M_2) < \eps$ by construction, $\lra$ reasoning proves the desired formula. 
\end{proof}

The completeness of topologically open properties (safety, liveness) then follows from the computability and existence of admissible approximations. Intuitively, suppose $O(x)$ defines an open set such that $\phi(x_0, t) \in \eval{O}$ for all $(x_0, t) \in \eval{C} \times [0, T]$. As $\phi(\eval{C}, [0, T])$ is a compact set, this necessarily implies the existence of some $\eps \in \Q^+$ such that the distance from $\phi(\eval{C}, [0, T])$ to the complement of $O(x)$ is at least $\eps$. Consequently, a sufficiently accurate admissible approximation $\Phi_k(x_0, t)$ will then witness this containment, thereby establishing that $O(x)$ contains the reachable set. 
\begin{proposition}[Density of Admissible Approximations]
    \label{prop: admissible approximations are dense}
    A computable sequence of admissible approximations exists. 
\end{proposition}
\begin{proof}
    This follows from the compactness of initial states $\eval{C}$. For simplicity, choose $M\in\Q^+$ large enough such that $\eval{C}\subseteq[-M,M]^n$ and abbreviate $G(x,t)\coloneqq e^{At}x$. For every $k\in\N^+$, construct a mesh-grid on $[-M,M]^n\times[0,T]$ with step sizes $M/k$, $T/k$ respectively. 
    At every mesh vertex $z=(q,t_i)$, compute a rational vector $y_z$
    satisfying
    \[
    \|y_z-G(q,t_i)\|_\infty<\frac{1}{k^2},
    \]
    Subdivide each mesh-grid cell into a fixed rational
    simplicial decomposition, and let $\Phi_k$ be the continuous
    function obtained by affine interpolation on each simplex, which is $\lra$ definable. 
    Since $G$ is smooth and the vertices have error $O(1/k^2)$, the constructed piecewise-affine interpolation on
    these increasingly fine meshes converges to $G$ uniformly together
    with its Lipschitz errors:
    \[
    \sup_{x_0\in \eval{C}}
    \|\Phi_k(x_0,t)-G(x_0,t)\|_{C^{0,1}}
    \longrightarrow 0.
    \]
    Exactly as desired. Furthermore, this construction is computable as affine interpolation is computable, completing the construction of a computable sequence of admissible approximations. 
\end{proof}
Proposition \ref{prop: admissible approximations are dense} guarantees the (computable) existence of admissible approximations. Together with \rref{thm:complete-convergence}, this implies the completeness of topologically open properties in the \emph{linear} fragment of \dL. The proofs are supplied in \rref{app:proofs}.
\begin{theorem}[Completeness of Open Safety]
\label{thm:complete-safety}
    Let $O(x) \in \lra$ define a topologically open subset of $\R^n$, then the linear fragment of $\dL$ is complete for the following:
    \[C(x) \land t = 0 \rightarrow \dbox{x' = Ax, t' = 1 \& t \leq T}{O(x)}\]
\end{theorem}
\begin{theorem}[Completeness of Open Liveness]
\label{thm:complete-live}
    Let $O(x) \in \lra$ define a topologically open subset of $\R^n$, then the linear fragment of $\dL$ is complete for the following:
    \[C(x) \land t = 0 \rightarrow \ddiamond{x' = Ax, t' = 1 \& t \leq T}{O(x)}\]
\end{theorem}
This completes the proof of completeness of (the linear fragment of) \dL for linear ODEs with respect to convergence, open safety and open liveness properties over finite time horizons. Proving that by specializing to the linear fragment, linearity can be preserved down to linear arithmetic while still achieving completeness properties. 

\section{Proving Stepwise Approximations}
\label{section: step-wise certification}
As alluded to in the introduction, while many algorithms have been developed to compute over-approximations of linear ODEs \cite{DBLP:conf/hybrid/Girard05, DBLP:conf/hybrid/GirardGM06, DBLP:conf/hybrid/WetzlingerKBA23}, there has been no unifying logical framework to generate corresponding proofs of correctness. This section justifies that \dL is a suitable axiomatic system for proving correctness of reachability computations from a practical point of view. It is worth noting that directly applying Theorem \ref{thm:complete-convergence} to certify the accuracy of numerical approximations is likely undesirable due to the following observations:
\begin{enumerate}
    \item Given a single $\lra$ definable approximation $\Phi_k(x_0, t)$ of error $h$, Lemma \ref{lem:linear-invariant} essentially establishes that $\Phi_k(x_0, t)$ has provable error $h((2 + T)e^{\normi{A}T} - 1)$, an exponential increase dependent on $T$.
    \item While the results in Section \ref{section: linear axiomatization} establish the completeness of the linear fragment of \dL for (open properties of) linear ODEs via density of admissible approximations, it does not guarantee that such admissible approximations can always be found efficiently.
\end{enumerate}
Furthermore, practical reachability computations of linear homogeneous ODEs $x' = Ax$ frequently utilize the homogeneity of linear ODEs to simplify computations. We consider the standard set-up for a reachability algorithm below \cite{DBLP:conf/hybrid/Girard05}:
\begin{enumerate}
    \item Decompose $[0, T]$ as $[0, N\delta]$ for some sufficiently small time step $\delta$, and compute an approximation $\Phi \approx e^{A\delta}$.
    \item Compute an over-approximation $Z_0$ of the reachable set on $[0, \delta]$ from the set of initial conditions $C$: 
        \[\{e^{As}x_0 \mid x_0 \in C \land s \in [0, \delta]\}\]
    In this section, we consider a standard approach where $Z_0$ is obtained as the bloated convex-hull of the sets $C, \Phi C$. Intuitively, this is just the image of the approximation:
    \[ (x_0, t) \mapsto x_0 + \frac{(\Phi - I)t}{\delta}x_0 \approx e^{At}x_0\] 
    \item As $Z_0$ over-approximates the flow on $[0, \delta]$, it follows that $Z_i \coloneqq \Phi^iZ_0$ (approximately) over-approximates the flow on $[i\delta, (i + 1)\delta]$, thus an overall approximation of the reachable set on $[0, T] = [0, N\delta]$ can be approximated as:
        \[\bigcup_{0 \leq i \leq N - 1} \Phi^iZ_0 = \bigcup_{0 \leq i \leq N - 1}Z_i\]
\end{enumerate}

The first result of this section is Theorem \ref{thm:step_approx}, establishing that \dL can carry out a similar algorithm by going slightly beyond the linear fragment, which is necessary to express the approximation $(x_0, t) \mapsto x_0 + \frac{(\Phi - I)t}{\delta}x_0$. Importantly, such a proof-producing algorithm in \dL does not require reasoning about nonlinear differential invariants via generic completeness theorems \cite{DBLP:journals/jacm/PlatzerT20}. 

\begin{theorem}[Stepwise Approximations]
    \label{thm:step_approx}
    Given an ODE $x' = Ax$ on a rational interval $[0, T]$ with $\lra$ definable compact initial conditions $C(x_0)$ and a desired number of time steps $N$, there computably exists: 
    \begin{itemize}
        \item Time varying matrices $D_i : [0, \delta] \to \R^{n \times n}$ for $i = 0, 1, \dots, N - 1$. 
        \item Error bounds $\eps_i$ for $i = 0, 1, \dots N - 1$. 
    \end{itemize}
    Such that the following $\dL$ formula is provable, where $\delta \coloneqq T/N$.
    \begin{align*}
        C(x_0) \land x = x_0 \land t = 0 &\longrightarrow\\
        \dbox{x' = Ax, t' = 1 \& t \leq T}{&\left(\bigwedge_{0 \leq i \leq N - 1} t \in [i\delta, (i + 1)\delta] \rightarrow \normi{x - D_i(t - i\delta)x_0} \leq \eps_i\right)}
    \end{align*}
    Furthermore, each $D_i(t)$ is affine in $t$ with rational coefficients and has global error $\max_{0 \leq i < N} \eps_i = O(\delta^2)$. 
\end{theorem}

Theorem \ref{thm:step_approx} can be viewed as the proof-theoretic analog of a standard reachability enclosure algorithm for linear systems \cite{DBLP:conf/hybrid/Girard05}. In particular, Theorem \ref{thm:step_approx} is constructive and only requires a number of \dL proof steps linear in $N$, and the global error bound of $O(\delta^2)$ cannot be improved in general since the initial approximation on $[0, \delta]$ is affine and thereby already has $O(\delta^2)$ error in the worst case. The $\dL$-provability part of the proof is presented below, and a complete argument for the error analysis is provided in \rref{app:proofs}.

\begin{proof}[Proof Sketch of Theorem \ref{thm:step_approx}]
    Intuitively, there are two parts to the algorithm: 
    \begin{itemize}
        \item Base case: Provably establishing a bound of the following form on $[0, \delta]$ for some $D_0 : [0, \delta] \to \Q^{n \times n}$:
            \[C(x_0) \land x = x_0 \land t = 0 \rightarrow \dbox{x' = Ax, t' = 1 \& t \leq \delta}{\normi{x - D_0(t)x_0} \leq \eps_0}\]
        \item Inductive case: Assuming that the bound has been established up to $[0, i\delta]$, extend this to $[0, (i + 1)\delta]$. 
    \end{itemize}
    For the base case, a sufficiently accurate Taylor approximation of $e^{At}$ suffices, which can be axiomatically proven in $\dL$ via standard arguments \cite{compact_IVP}. The base case is also the only part of the proof where nonlinear real arithmetic is used, to obtain error bounds on Taylor approximations. For the inductive step, we first make the following abbreviations (where $I_i = [i\delta, (i + 1)\delta]$):
    \begin{align*}
        \Gamma &\equiv C(x_0) \land x = x_0 \land t = 0\\
        \Acc_i(\eps)
        &\equiv
        t\in I_i\rightarrow
        \normi{x-D_i(t-i\delta)x_0}\leq\eps\\
        \Rad_i(M, \eps)
        &\equiv
        s \in [0, \delta] \land C(x_0) \rightarrow \normi{D_i(s)x_0} + \eps\leq M
    \end{align*}
    The derivation also uses the following end-point error bound for some sufficiently accurate $D \approx e^{A\delta}$, which can similarly be proven by Taylor approximations \footnote{Technically this uses nonlinear arithmetic. However, one Taylor approximation will suffice for both $\End(\eta)$ and the base case.}. 
    \[\End(\eta)
        \equiv
        \normi{x} \leq 1 \land x_0 = x \land t = 0 \rightarrow
        \dbox{x'=Ax,t'=1{\&}t\leq\delta}
        {\bigl(t=\delta\rightarrow\normi{x-Dx_0}\leq \eta\bigr)}\]
    Intuitively, $\Acc_i(\eps)$ encodes that the error bound is at most $\eps$ at the $i$-th time step, $\Rad_i(M, \eps)$ bounds the maximum norm of the $i$-th approximation and $\End(\eta)$ bounds $\normi{e^{A\delta} - D}$. It then suffices to establish that the following proof rule is derivable in $\dL$, where $D_i \coloneqq DD_{i - 1}$ is the previous approximation multiplied by $D$.
    \begin{center}
    \resizebox{0.98\linewidth}{!}{$\displaystyle
    \textsc{Step}_{\delta}\quad
    \linferenceRule[sequent]
    {
        \lsequent{\Gamma}
        {\dbox{\flow{i\delta}}
            {\Acc_{i-1}(\eps)}}
        &
        \lsequent{}{\Rad_{i - 1}(M, \eps)}
        &
        \lsequent{}{\End(\eta)}
        &
        \lsequent{}{\eps^+\geq\normi{D}\eps+M\eta}
    }
    {\lsequent{\Gamma}
        {\dbox{\flow{(i+1)\delta}}{\Acc_i(\eps^+)}}}
    $}
    \end{center}
     Indeed, if \textnormal{\textsc{Step}}$_\delta$ is valid, then $(N - 1)$ repeated applications on the base case will prove the desired bound on $[0, N\delta]$ with $\eps_{i} = \normi{D}\eps_{i - 1} + M_{i - 1}\eta$. Its derivation can be established as follows:
    \begin{enumerate}
    \item This is proved by contradiction. Suppose that the bound does not hold, hence the following $\dL$ formula is satisfiable, i.e. a state of failure can be reached. 
        \begin{align*}
            \Gamma \land \ddiamond{x' = Ax, t' = 1 \& t \leq (i + 1)\delta}{\left(t \in I_i \land \normi{x - D_{i}(t - i\delta)x_0} > \eps_{i}\right)}
        \end{align*}
    \item Denote $r \in I_i$ to be the value of the time variable $t$ reached at the failing state. By the intermediate value theorem in $\dL$ \cite[Lemma~15]{Platzer_Tan_report} applied to $t - (r - \delta)$, the flow to the failing state will first reach some state with $t = r - \delta \in [(i - 1)\delta, i \delta]$. By our premise, at this state one has $\Acc_{i-1}(\eps_{i - 1})$, so $x = D_{i - 1}(r - i\delta)x_0 + \vec{\gamma}$ where $\vec{\gamma}$ is some error term with $\normi{\vec{\gamma}} \leq \eps_{i - 1}$. 
    \item Note that the flow of $D_{i - 1}(r - i\delta)x_0$ after a duration of $t = \delta$ semantically satisfies the following inequality (using the premise of $\Rad_{i - 1}(M, \eps)$)
    \begin{align*}
        \normi{e^{A\delta}D_{i - 1}(r - i\delta)x_0 - DD_{i - 1}(r - i\delta)x_0} &\leq \normi{e^{A\delta} - D}\normi{D_{i - 1}(r - i\delta)x_0}\\
        &\leq \normi{e^{A\delta} - D}(M - \eps)
    \end{align*}
    By the premise of $\End(\eta)$, (semantically) this can further be reduced to an upper bound of $(M - \eps)\eta$. Crucially, this semantic bound is actually syntactically provable, i.e. the following $\dL$ formula is provable
    \begin{align*}
        &x = D_{i - 1}(r - i\delta)x_0 \land t = r - \delta \land \normi{D_{i - 1}(r - i\delta)x_0} \leq M - \eps \longrightarrow\\
        &\dbox{\flow{(i + 1)\delta}}{(t = r \rightarrow \normi{x - DD_{i - 1}(r - i\delta)x_0} \leq (M - \eps)\eta)}
    \end{align*}
    Indeed, this is achieved by applying the provable superposition principle (\rref{lem:superposition}) with the scaling $x \to \frac{x}{M - \eps}$ so that $\normi{x} \leq 1$ is satisfied and the premise $\End(\eta)$ can be applied. A similar computation also establishes that the flow of $\vec{\gamma}$ for $\delta$ time can be bounded by $\normi{D}\eps_{i - 1} + \eps\eta$. 
    \item Thus, the (provable) $i$-th approximation bound can be taken to be $D_{i}(t) = DD_{i - 1}(t)$ and $\eps_i = \normi{D}\eps_{i - 1} + \eps\eta + (M - \eps)\eta = \normi{D}\eps_{i - 1} + M\eta$, deriving the proof rule \textnormal{\textsc{Step}}$_\delta$. 
\end{enumerate}

\end{proof}
The following example provides one such \emph{provable approximation} of a random $3$-dimensional linear ODE. 
\begin{example}[3D Example]
    We compute the reachability enclosure of a random ODE $x' = Ax$ using the algorithm outlined in Theorem \ref{thm:step_approx}, on $[0, 1]$ with a time step of $\delta = \frac{1}{20}$, the matrix $A$ being
    \[
  A =
  \begin{pmatrix}
    -\frac{1}{2} & \frac{3}{10} & -\frac{1}{5} \\
    -\frac{1}{5} & -\frac{2}{5} & \frac{1}{10} \\
    \frac{1}{10} & -\frac{1}{5} & -\frac{1}{2}
  \end{pmatrix}.
  \]
  And the initial conditions $C = [-1/20, 1/20]^3$. This approximation is shown in Figure \ref{fig:box-flowpipe}, which has provable error $\approx 10^{-3}$ certified by KeYmaera X. 
\begin{figure}
    \centering
    \includegraphics[
      width=0.5\linewidth
    ]{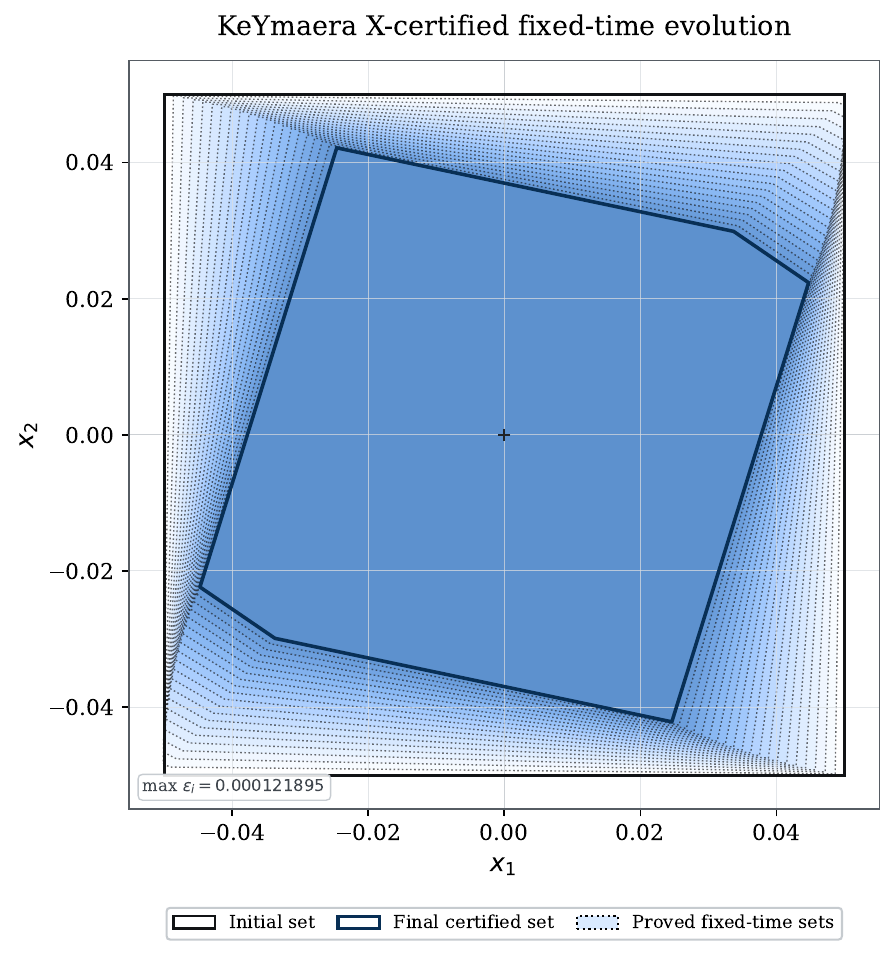}
    \caption{Certified evolution projected onto $(x_1, x_2)$ over $[0, 1]$}
    \label{fig:box-flowpipe}
  \end{figure}
\end{example}

Finally, a direct application of \rref{thm:step_approx} is that $\dL$ is capable of directly certifying reachable sets computed by external tools. 

\begin{corollary}
    \label{cor:step_cert}
    Let $x' = Ax, [0, T], C(x_0), N, \delta$ be as given in \rref{thm:step_approx}. Further assume that a sequence of (bounded) $\lra$ definable sets $Z_0, \dots, Z_{N - 1}$ and a corresponding $\eps > 0$ are given satisfying:
    \begin{itemize}
        \item $B_\infty(\phi(C, [0, \delta]), \eps) \subseteq Z_0$
        \item $B_\infty(\phi(Z_i, \delta), \eps) \subseteq Z_{i + 1} \qquad (0 \leq i < N - 1)$
    \end{itemize}
    where $\phi(x_0, t) = e^{At}x_0$ denotes the flow and $B_\infty(A, r)$ denotes the infinity-norm neighborhood of radius $r$ around $A$. Then the following $\dL$ formula is directly provable, certifying the correctness of $\lra$ approximations:
    \begin{align*}
        C(x_0) \land x = x_0 \land t = 0 \rightarrow \dbox{x' = Ax, t' = 1 \& t \leq T}{\left(\bigwedge_{0 \leq i \leq N - 1} t \in [i\delta, (i + 1)\delta] \rightarrow Z_i(x)\right)}
    \end{align*}
\end{corollary}
Intuitively, \rref{cor:step_cert} establishes that $\dL$ can directly certify the correctness of reachability sets provided that they allow for a positive numerical tolerance, which is expected of numerical approximations. The proof is supplied in \rref{app:proofs}. An important feature of the proof is that the nonlinear fragment of $\dL$ is only needed to establish the correctness of $Z_0$, and all containments of the form $\phi(Z_i, \delta) \subseteq Z_{i + 1}$ are soundly reduced to $\lra$. 

\begin{example}[Certified Zonotopes]
    As an example application of \rref{cor:step_cert}, we certify a sequence of 100 time-step reachability enclosures of the ODE $x' = Ax$ over the time interval $[0, 1]$ using $\delta = \frac{1}{100}$ and initial conditions $[0.9, 1.1]^5$ with 
    \[
  A =
  \begin{pmatrix}
  -1 & -4 &  0 &  0 &  0 \\
   4 & -1 &  0 &  0 &  0 \\
   0 &  0 & -3 &  1 &  0 \\
   0 &  0 & -1 & -3 &  0 \\
   0 &  0 &  0 &  0 & -2
  \end{pmatrix}.
  \]
  The set approximations were zonotopes computed by the external tool JuliaReach \cite{DBLP:conf/hybrid/BogomolovFFPS19}, with a conservative rounding ($\approx 10^{-5}$) applied to account for floating-point errors. The projections of the certified reachable sets are shown in \rref{fig:5d-flowpipe}. It should be stressed that these reachable sets are \emph{provably correct} as certified by KeYmaera X, which were obtained via \emph{unverified} approximations from external tools.
  \begin{figure}
    \centering
    \includegraphics[
      width=1.0\linewidth
    ]{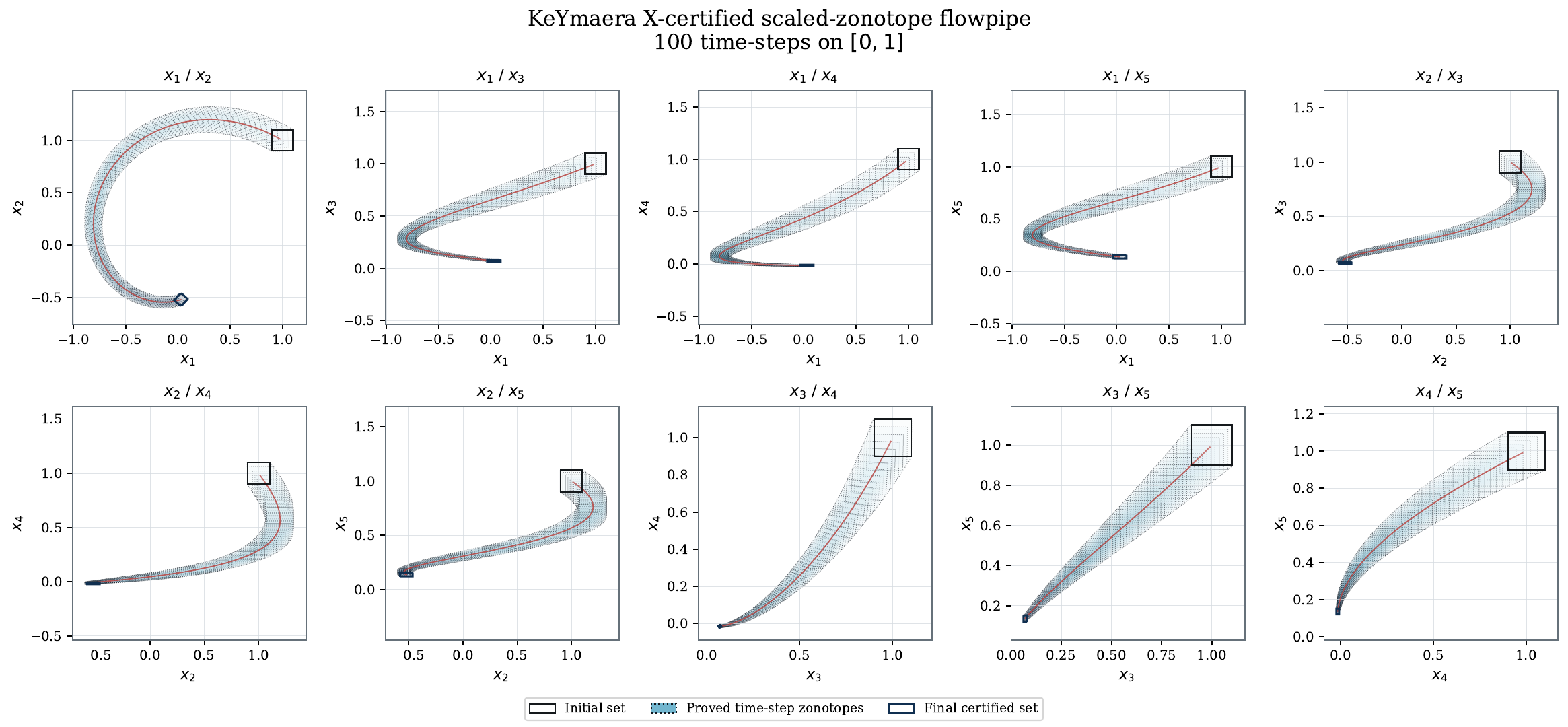}
    \caption{Certified evolution of $5$D system over $[0, 1]$}
    \label{fig:5d-flowpipe}
  \end{figure}
\end{example}

\section{Related Work}
\noindent\textbf{Axiomatization of Differential Equations.} This article contributes to the axiomatic study of (linear) ODEs with $\dL$. Earlier works in $\dL$ either do not establish (absolute) completeness \cite{DBLP:conf/lics/Platzer12b, DBLP:conf/lics/Platzer12a, DBLP:journals/jar/Platzer08}, or prove completeness for fragments of $\dL$ that include general polynomial ODEs \cite{DBLP:journals/jacm/PlatzerT20, compact_IVP} and therefore require the use of nonlinear real arithmetic. In particular, earlier work \cite{compact_IVP} concerning the completeness of open properties on finite time horizons relies upon the general completeness of (nonlinear) differential invariants. Beyond $\dL$, axiomatic approaches to continuous/hybrid dynamical systems include automata-theoretic approaches \cite{DBLP:conf/hybrid/NerodeK92a,DBLP:conf/hybrid/AlurCHH92}, hybrid Hoare logic \cite{DBLP:conf/fm/ShengBZ23} and embeddings such as IsaVODEs \cite{DBLP:journals/jar/MuniveFGSLH24}. The general decidability and completeness of differential invariants in $\dL$ have been established in earlier works \cite{DBLP:conf/emsoft/LiuZZ11,DBLP:journals/jsc/GhorbalS22, DBLP:journals/jacm/PlatzerT20}, which this article expands upon. The specific (arithmetic) problem of certifying zonotope containment in $\dL$ has also been explored in earlier work \cite{DBLP:conf/ifm/HellwigSQPA25}.

\noindent\textbf{Linear Differential Equations.} The quantitative study of linear ODEs, specifically in the case of safety verification, has resulted in many efficient algorithms to compute over-approximations of reachable sets, implemented in tools such as CORA \cite{DBLP:conf/cpsweek/Althoff15}, Flow* \cite{DBLP:conf/cav/ChenAS13}, SpaceEx \cite{DBLP:conf/cav/FrehseGDCRLRGDM11}, JuliaReach \cite{DBLP:conf/hybrid/BogomolovFFPS19} and many others. While such algorithms are mathematically rigorous in theory, most of them have not been formally verified outside of a few exceptions \cite{DBLP:journals/jar/Immler18, DBLP:conf/itp/0001T24}. In particular, such approaches focus on formally verifying \emph{algorithms} for reachability computations, as opposed to the approach presented in this article which certifies \emph{results} with proofs in $\dL$.

It should be stressed that while linear ODEs are mathematically ``simple'' in the sense that a closed-form solution $\phi(x_0, t) = e^{At}x_0$ can be written down explicitly, designing provably correct and scalable algorithms to over-approximate reachable sets is far from trivial. Even the computation of the underlying matrix exponential is surprisingly intricate \cite{DBLP:journals/siamrev/MolerL03}. Traditional algorithms for numerically solving ODEs (e.g., Runge-Kutta methods) are primarily designed to compute the approximate reachable set starting from a \emph{single} initial condition instead of a set of initial conditions as needed for the analysis of cyber-physical systems, and generally do not provide guaranteed error bounds that are needed for safety verification.

Concerning the qualitative study of linear ODEs, many deep results have been established that relate computability-theoretic questions to open problems in number theory \cite{DBLP:journals/tcs/BellDJB10, DBLP:conf/icalp/ChonevOW16}. Results concerning the general synthesis of o-minimal definable differential invariants focusing on infinite time horizons avoiding semi-algebraic unsafe sets have also been established \cite{DBLP:conf/icalp/AlmagorKO020}.

\section{Conclusion}
This article establishes that the linear fragment of \dL is complete for topologically open properties of linear ODEs, in particular including correctness properties of (numerical) reachability approximations. Topologically open properties of linear ODEs can be proven in linear real arithmetic with a linear number of Lie derivatives, establishing the theoretical foundations for complete proof-producing certifications of reachability computations for linear ODEs.
Furthermore, this article establishes that a standard algorithm \cite{DBLP:conf/hybrid/Girard05} can be provably implemented in full (nonlinear) \dL using a linear number of proof steps. For future work, it would be interesting to determine the computational complexity of such certification procedures. 

\bibliographystyle{plainurl}
\bibliography{references}

\appendix

\section{Axiomatization}
\label{app:axioms}
This section provides a complete record of $\dL$'s axiomatization for reference.

\begin{theorem}[\cite{DBLP:conf/lics/Platzer12b, DBLP:journals/jacm/PlatzerT20, DBLP:journals/fac/TanP21}]
    \label{thm: base axiomatization of dL}
    The following are sound axioms of $\dL$. In axioms \irref{cont}, \irref{dadj}, \irref{bdg}, the variables $y$ is fresh. In axiom \irref{bdg}, $Q(x)$ is required to be a formula of real arithmetic. 

    \begin{calculus}
        \cinferenceRule[qear|\usebox{\Rval}]{quantifier elimination real arithmetic}
        {\linferenceRule[sequent]
          {}
          {\lsequent[g]{\Gamma}{\Delta}}
        }{$\text{if}~\landfold_{\ausfml\in\Gamma} \ausfml \limply \lorfold_{\busfml\in\Delta} \busfml ~\text{is valid in \LOS[\reals]}$}%
    
        \cinferenceRule[diamond|$\didia{\cdot}$]{diamond axiom}
        {\linferenceRule[equiv]
          {\lnot\dbox{\ausprg}{\lnot \ausfml}}
          {\ddiamond{\ausprg}{\ausfml}}
        }
        {}

        \cinferenceRule[evolved|$\didia{'}$]{evolve}
        {\linferenceRule[equiv]
          {\lexists{t{\geq}0}{\ddiamond{\pupdate{\pumod{x}{y(t)}}}{p(x)}}\hspace{1cm}}
          {\ddiamond{\pevolve{\D{x}=\genDE{x}}}{p(x)}}
        }{$\m{\D{y}(t)=\genDE{y}}$}%

        \cinferenceRule[B|B$'$]{}
        {\linferenceRule[equiv]
          {\lexists{y}{\ddiamond{\pevolvein{x'=f(x)}{Q(x)}}{\rfvar(x,y)}}}
          {\ddiamond{\pevolvein{x'=f(x)}{Q(x)}}{\exists{y}\rfvar(x,y)}}
        }{\text{$y \not\in x$}}
        
        \cinferenceRule[K|K]{K axiom / modal modus ponens} %
        {\linferenceRule[impl]
          {\dbox{\alpha}{(\fvarA \limply \fvarB)}}
          {(\dbox{\alpha}{\fvarA}\limply\dbox{\alpha}{\fvarB})}
        }{}
        \cinferenceRule[V|V]{vacuous $\dbox{}{}$}
         {\linferenceRule[impl]
           {\fvarA}
           {\dbox{\alpha}{\fvarA}}
         }{\text{no free variable of $\fvarA$ is bound by $\alpha$}}
        \cinferenceRule[G|G]{$\dbox{}{}$ generalization} %
        {\linferenceRule[formula]
          {\lsequent{}{\fvarA}}
          {\lsequent{\Gamma}{\dbox{\alpha}{\fvarA}}}
        }{}
        
        \cinferenceRule[dW|dW]{}
        {\linferenceRule
          {\lsequent{\ivr}{P}}
          {\lsequent{\Gamma}{\dbox{\pevolvein{\D{x}=\genDE{x}}{\ivr}}{P}}}
        }{}
        
        \cinferenceRule[dC|dC]{differential cut}%
        {\linferenceRule[sequent]
          {\lsequent[L]{}{\dbox{\pevolvein{\D{x}=\genDE{x}}{\ivr}}{\cusfml}}
          &\lsequent[L]{}{\dbox{\pevolvein{\D{x}=\genDE{x}}{(\ivr\land \cusfml)}}{\ousfml[x]}}}
          {\lsequent[L]{}{\dbox{\pevolvein{\D{x}=\genDE{x}}{\ivr}}{\ousfml[x]}}}
        }{}
        
        \cinferenceRule[DG|DG]{differential ghost variables}
        {\linferenceRule[viuqe]
          {\dbox{\pevolvein{x'=\genDE{x}}{\ivr(x)}}{\ousfml[x](x)}}
          {\lexists{y}{\dbox{\pevolvein{x'=\genDE{x}\syssep y'=a(x)\cdot y+b(x)}{\ivr(x)}}{\ousfml[x](x)}}}
        }
        {}

        \cinferenceRule[DGi|DGi]{differential ghost variables}
        {\linferenceRule[impl]
          {\dbox{\pevolvein{x'=\genDE{x}}{\ivr(x)}}{\ousfml[x](x)}}
          {\forall{y}{\dbox{\pevolvein{x'=\genDE{x}\syssep y'= g(x, y)}{\ivr(x)}}{\ousfml[x](x)}}}
        }
        {}
        
        \cinferenceRule[thereAndBack_orig|{[$\&$]}]{there and back quantification}
       {
       \linferenceRule[equiv]
       {\forall t_0 {=} c_0 \dbox{x' = \theta}{\left(\dbox{x' = -\theta}{\left(c_0 \geq t_0 \rightarrow \chi\right) \rightarrow \phi}\right)}}
        {\dbox{x' = \theta \& \chi}{\phi}}
       }
       {}

        \cinferenceRule[dx|DX]{differential skip}
        {
            \linferenceRule[equiv]
            {\left(Q \rightarrow P \land \dbox{x' = f(x) \& Q}{P}\right)}
            {\dbox{x'= f(x) \& Q}{P}}
        }
        {$x' \notin P, Q$}
        
        \cinferenceRule[uniq|Uniq]{vanilla uniqueness}
        {\linferenceRule[equiv]
        {\left(\ddiamond{x'= f(x) \& Q_1}{P}\right) \land \left(\ddiamond{x'= f(x) \& Q_2}{P}\right)}
        {\ddiamond{x'= f(x) \& Q_1 \land Q_2}{P}}
        }{}

        \cinferenceRule[cont|Cont]{continuity of ODE}
        {\linferenceRule[impl]
        {x = y}
        {\left(\ddiamond{x'= f(x) \& e > 0}{x \neq y} \leftrightarrow e > 0\right)}
        }{$f(x) \neq 0$}

        \cinferenceRule[dadj|Dadj]{reverse flow of ODEs}
        {\linferenceRule[equiv]
        {\ddiamond{y' = -f(y) \& Q(y)}{y = x}}
        {\ddiamond{x' = f(x) \& Q(x)}{x = y}}
        }{}

        \cinferenceRule[ri|RI]{real induction axiom}
        {\linferenceRule[equiv]
        {\forall y\dbox{x' = f(x) \& P \lor x = y}{\left(x = y \rightarrow P \land \ddiamond{x' = f(x) \& P \lor x = y}{x \neq y}\right)}}
        {\dbox{x' = f(x)}{P}}
        }{}
        
        \cinferenceRule[bdg|BDG]{bounded differential ghost}
       {
       \linferenceRule[impll]
       {\dbox{x' = f(x), y' = g(x, y) \& Q(x)}{\norm{y}^2 \leq p(x)}}
        {\left(\dbox{x' = f(x) \& Q(x)}{P(x)} \leftrightarrow \dbox{x' = f(x), y' = g(x, y) \& Q(x)}{P(x)}\right)}
       }
       {}
       
    \end{calculus}
\end{theorem}

\begin{remark}
    In axioms \irref{thereAndBack_orig} and \irref{cont}, it is assumed that the ODE $x' = f(x)$ includes a clock variable $c_0' = 1$. This assumption can be made without loss of generality since a clock variable can always be added using \irref{DG}. The variable $t_0$ is also assumed to be fresh in \irref{thereAndBack_orig}.
\end{remark}

\section{Proofs}
\label{app:proofs}

\begin{proof}[Proof of \rref{thm:complete-safety}]
    To establish completeness, we assume the validity of the formula and aim to show provability. As $O(x)$ is open and $\eval{C}$ is compact, there must exist some $\eps \in \Q^+$ such that $d(\phi(\eval{C}, [0, T]), \eval{O}^c) > \eps$, i.e. the flow is of distance at least $\eps$ to the complement of $\eval{O}$. Let $\eps \in \Q^+$ be some such witness, and $\Phi(x_0, t)$ be some $\lra$ definable approximation of $\phi(x_0, t)$ with \emph{provable} error at most $\frac{\eps}{4}$, which necessarily exists by Proposition \ref{prop: admissible approximations are dense} and Theorem \ref{thm:complete-convergence}. The following $\lra$ formula is therefore valid and provable
    \begin{align*}
        &C(x_0) \land t \in [0, T] \land \normi{\Phi(x_0, t) - y} < \frac{\eps}{4} \rightarrow O(y)
    \end{align*}
    As $\Phi(x_0, t)$ was constructed to have provable error at most $\frac{\eps}{4}$, this is equivalent to the provability of 
    \begin{align*}
        &C(x_0) \land x = x_0 \land t = 0 \rightarrow \dbox{x' = Ax, t' = 1 \& t \leq T}{ \normi{\Phi(x_0, t) - x} < \frac{\eps}{4}}
    \end{align*}
    Combining these then establishes the provability of the following as desired
    \[C(x)\land t = 0 \rightarrow \dbox{x' = Ax, t' = 1 \& t \leq T}O(x) \qedhere\]
\end{proof}

\begin{proof}[Proof of \rref{thm:complete-live}]
    The proof of liveness is similar to the proof of Theorem \ref{thm:complete-safety}. Recall that the liveness formula is equivalent to the validity of the following quantified formula involving the flow $\phi(x_0, t)$.
    \[\forall x_0 \in \eval{C} \exists s \in [0, T]\, O(\phi(x_0, s))\]
    Clearly, a sufficient condition of provability is if there exists some approximant $\Phi(x_0, t)$ with provable error $\eps \in \Q^+$ such that the following $\lra$ formula is valid 
    \[\forall x_0 \in \eval{C} \exists s \in [0, T] \forall y \normi{\Phi(x_0, s) - y} < \eps \rightarrow O(y)\]
    I.e. For every initial condition $x_0$, there exists some time $s \in [0, T]$ such that the $\eps$-ball around $\Phi(x_0, s)$ lies within the target set $\eval{O}$. However, such an approximant must necessarily exist as $\eval{O}$ is open. Indeed, suppose that this is false, then there exists a sequence of approximations $\Phi_k(x_0, s)$ converging to $\phi(x_0, s)$ with uniform error $\normi{\Phi_k(x_0, s) - \phi(x_0, s)} < 2^{-k}$, such that for all $k \in \N$, there exists $x_k \in \eval{C}$ where the distance to the complement of $\eval{O}$ satisfies $d(\eval{O}^c, \Phi_k(x_k, s)) < 2^{-k}$ for all $s \in [0, T]$. However, the compactness of $\eval{C}$ then implies the existence of a limiting point (up to taking a subsequence) $x_k \to x^* \in \eval{C}$, satisfying $d(\eval{O}^c, \phi(x^*, s)) = 0$ for all $s \in [0, T]$. But since $\eval{O}^c$ is closed, this implies that $x^*$ never reaches $\eval{O}$, a contradiction. 
\end{proof}

The following lemma provably establishes the superposition principle in $\dL$
\begin{lemma}[Linear Superposition]
    \label{lem:superposition}
    The following is derivable axiom of $\dL$
    \[y = ax + z \rightarrow \dbox{x' = Ax, y' = Ay, z' = Az}{y = ax + z}\]
\end{lemma}
\begin{proof}
    Taking the Lie derivative of the difference gives
    \[(ax + z - y)' = aAx + Az - Ay = A(ax + z - y)\]
    Hence $ax + z = y$ forms a Darboux invariant \cite{DBLP:journals/jacm/PlatzerT20} of the system, and is therefore provable. 
\end{proof}

\begin{proof}[Proof of \rref{thm:step_approx}]
    Let $C(x_0), x' = Ax, [0, T], N$ be as given in the theorem. Let 
    \[P_m(t) = \sum_{i = 0}^m \frac{A^i}{i!}t^i\]
    be the $m$-th order Taylor approximation of $e^{At}$. By completeness of \dL for Taylor approximations \cite{compact_IVP} and the standard Taylor remainder theorem, there exists some $\eps : \N^+ \to \Q^+$ such that $\eps(m) = O(\delta^{m + 1})$ and the following is provable for all $m$:
    \[\normi{x} \leq 1 \land x = x_0 \land t = 0 \rightarrow \dbox{x' = Ax, t' = 1 \& t \leq \delta}{\normi{x - P_m(t)x_0} < \eps(m)}\]
    Let $D \approx e^{A\delta}$ be any rational approximation with error $O(\delta^3)$, let $D_t \coloneqq (I + \frac{D - I}{\delta}t)$. Then there exists $\beta = O(\delta^2), \eta = O(\delta^3)$ such that the following are provable (using sufficiently accurate Taylor approximations)
    \begin{align*}
        &\normi{x} \leq 1 \land x = x_0 \land t = 0 \rightarrow \dbox{x' = Ax, t' = 1 \& t \leq \delta}{\normi{x - D_tx_0} < \beta}\\
        &\normi{x} \leq 1 \land x = x_0 \land t = 0 \rightarrow \dbox{x' = Ax, t' = 1 \& t \leq \delta}{(t = \delta \rightarrow \normi{x - Dx_0} < \eta)}
    \end{align*}
    We now follow the proof sketch of \rref{thm:step_approx} and construct the sequence of affine matrices $D_i$ and error bounds $\eps_i$ via:
    \begin{itemize}
        \item $D_i(t) \coloneqq D^iD_t$. 
        \item $R \coloneqq \sup_{x \in C} \normi{x}$ and the radius bounds
         \[M_i \coloneqq \sup_{x \in D_i([0, \delta])C} \normi{x} + 1 \qquad\qquad M \coloneqq \max\left\{\max_{0 \leq i \leq N - 1} M_i, \sup_{x \in C} \normi{x}\right\}\]
        \item $\eps_{0} = R\beta, \eps_{i + 1} = \normi{D}\eps_i + M\eta$
    \end{itemize}
    The $\dL$ provability of these bounds follow directly from the sketch (assuming without loss of generality that $\eps_i < 1$), so it suffices to prove $\max_{0 \leq i \leq N - 1} \eps_i = O(\delta^2)$. Unfolding the recurrence gives
    \[
        \eps_i
        =
        \normi{D}^i\eps_0
        +
        M\eta\sum_{j=0}^{i-1}\normi{D}^j.
    \]
    The approximation on $\eta = O(\delta^3)$ gives:
    \[
        \normi{D - e^{A\delta}} = O(\delta^3) \implies \normi{D} = 1 + O(\delta)
    \] 
    Thus, the summation simplifies to
    \[
        \sum_{j=0}^{i-1}\normi{D}^j=O(N)=O(\delta^{-1}).
    \]
    Using $\eps_0=R\beta=O(\delta^2)$ and
    $\eta=O(\delta^3)$ therefore yields
    \[
        \max_{0\leq i<N}\eps_i
        =
        O(\beta+N\eta)
        =
        O(\delta^2).
    \]
    Establishing the global $O(\delta^2)$ error as desired. 
\end{proof}

\begin{proof}[Proof of \rref{cor:step_cert}]
    For $0 \leq i \leq N - 1$, let $M_i \coloneqq \sup_{x \in Z_i} \normi{x}$ and define:
    \[M \coloneqq \max\left\{\max_{0 \leq i \leq N - 1} M_i, \sup_{x \in C} \normi{x}\right\} \qquad\qquad\qquad\qquad \eta \coloneqq \frac{\eps}{4M}\]
    Similar to \rref{thm:step_approx}, we may establish some (possibly nonlinear) $\folr$ definable approximation (e.g. Taylor approximation) $D : [0, \delta] \to \R^{n \times n}$ such that the following is provable:
    \[\normi{x} \leq 1 \land x = x_0 \land t = 0 \rightarrow \dbox{x' = Ax, t' = 1 \& t \leq \delta}{\normi{x - D(t)x_0} \leq \eta}\]
    this is also the only place where the nonlinear fragment of $\dL$ is needed. The correctness of $Z_i$ can then be provably established as follows:
    \begin{itemize}
        \item Base case - This directly follows from the provable error of $D(t)$ and the (provable) superposition principle of linear ODEs. For $x \in C$, we may construct the auxiliary scaled variable $z = \frac{x}{M}$ satisfying $\normi{z} \leq 1$. Thus, by construction of $D(t)$, the following is provable
        \[C(x) \land x = x_0 \land z = \frac{x}{M} \land t = 0 \rightarrow \dbox{z' = Az, t' = 1 \& t \leq \delta}{\normi{z - \frac{D(t)x_0}{M}} \leq \eta}\]
        By the provable superposition principle of linear ODEs, it follows that $z = \frac{x}{M}$ along the ode $x' = Ax, z' = Az$. Thus, this establishes that $\normi{x - D(t)x_0} \leq M\eta$ for any initial condition $x_0 \in C$. By construction of $\eps$, the following arithmetic fact is necessarily valid
        \[t \in [0, \delta] \land C(x_0) \rightarrow B(D(t)x_0, \frac{3\eps}{4}) \subseteq Z_0\]
        Therefore this establishes the base case safety requirement:
        \[C(x) \land x = x_0 \land t = 0 \rightarrow \dbox{x' = Ax, t' = 1 \& t \leq \delta}{Z_0(x)}\]
        \item Inductive case - This is similar to the proof of \rref{thm:step_approx}, let $D \coloneqq D(\delta)$, then evolving from $t \mapsto t + \delta$ is exactly $x \mapsto e^{A\delta}x \approx Dx$. More formally, by a similar IVT argument to the proof of \rref{thm:step_approx}, the reachable set $Z_i$ is provably propagated to 
        \[Z_i \mapsto DZ_i + B(0, M_i\eta) \subseteq \phi(Z_i, \delta) + B(0, M_i\eta) + B(0, M_i\eta) \subseteq \phi(Z_i, \delta) + B(0, \eps) \subseteq Z_{i + 1}\]
        completing the inductive claim. In particular, the only arithmetic condition that needs to be proven is the containment 
            \[DZ_i + B(0, M_i\eta) \subseteq Z_{i + 1}\]
        which is in $\lra$. 
    \end{itemize}
\end{proof}

\end{document}